\documentclass{IEEEtran}

\usepackage[ruled,linesnumbered]{algorithm2e}
\usepackage{bbm}
\usepackage{tikz}
\usepackage{tkz-berge}
\newcommand{\h}{\hspace*{0.2in}}
\newcommand{\sh}{~~}
\usepackage{xcolor}
\usepackage{graphicx}
\usepackage{subcaption}

\usepackage[utf8]{inputenc}
\usepackage[english]{babel}
\usepackage{amssymb,amsmath,amsthm}
\usepackage{enumerate}

\newtheorem{definition}{Definition}
\newtheorem{theorem}{Theorem}
\newtheorem{conjecture}[theorem]{Conjecture}
\newtheorem{lemma}[theorem]{Lemma}

\title{NC Algorithms for Popular Matchings in One-Sided Preference Systems and Related Problems} 

\author{\IEEEauthorblockN{Changyong Hu, Vijay K. Garg} \\
\IEEEauthorblockA{\textit{Department of Electrical and Computer  Engineering} \\
\textit{University of Texas at Austin} \\
colinhu9@utexas.edu, garg@ece.utexas.edu}
}

\begin{document}

\maketitle

\begin{abstract}
The popular matching problem is of matching a set of applicants to a set of posts, where each applicant has a preference list, ranking a non-empty subset of posts in the order of preference, possibly with ties. A matching $M$ is popular if there is no other matching $M'$ such that more applicants prefer $M'$ to $M$. We give the first NC algorithm to solve the popular matching problem without ties. We also give an NC algorithm that solves the maximum-cardinality popular matching problem. No NC or RNC algorithms were known for the matching problem in preference systems prior to this work. Moreover, we give an NC algorithm for a weaker version of the stable matching problem, that is, the problem of finding the ``next" stable matching given a stable matching. 
\end{abstract}

\section{Introduction}
The matching problem in preference systems is well-studied in economics, mathematics and computer science, see for example \cite{roth1977weak, zhou1990conjecture, gale1962college}. It models many important real-world applications, including the allocation of medical residents to hospitals \cite{roth1984evolution} and families to government-owned housing \cite{abdulkadiroglu1998random}. The notion of a popular matching was first introduced in \cite{gardenfors1975match} in the context of the stable marriage problem. We say that a matching $M$ is {\em more popular than} $M'$ if the number of nodes that prefer $M$ to $M'$ exceeds the number of nodes that prefer $M'$ to $M$. A matching $M$ is popular if $M$ is optimal under the {\em more popular than} relation. Gupta et al. \cite{gupta2019popular} and Faenza et al. \cite{faenza2019popular} recently showed that the popular matching problem is NP-complete in the general roommate setting. The popular matching problem we consider is from \cite{abraham2007popular} such that the preference system is only one-sided. Abraham et al. \cite{abraham2007popular} gave a linear-time algorithm for the problem in the case of strictly-ordered preference lists and a polynomial-time algorithm for the case of preference lists with ties. There are other problems with other definitions of optimality such as {\em Pareto optimal} matching \cite{abraham2004pareto}, {\em rank-maximal} matching \cite{irving2004rank} etc. We do not discuss them here. 

The matching problem in the normal case, that is the problem of checking if a given graph has a perfect matching, and the corresponding search problem of finding a perfect matching have received considerable attention in the field of parallel computation. Tutte and Lovasz \cite{lovasz1979determinants} observed that there is an RNC algorithm for the decision problem. The search version was shown to be in RNC by Karp, Upfal and Wigderson \cite{karp1986constructing} and subsequently by Mulmuley, Vazirani and Vazirani using the celebrated Isolation Lemma \cite{mulmuley1987matching}. We note that no NC or RNC algorithms were known for the matching problem in preference systems prior to this work. The problem of finding an NC algorithm for the stable marriage problem has been open for a long time. Mayr and Subramanian \cite{mayr1992complexity} showed that the stable marriage problem is CC-complete. Subramanian \cite{subramanian1991computational} defined the complexity class CC as the set of problems log-space reducible to the comparator circuit value problem (CCV). Cook et al. \cite{cook2014complexity} conjectured that CC is incomparable with the parallel class NC, which implies none of the CC-complete problems has an efficient polylog time parallel algorithm. Recently, Zheng and Garg \cite{zheng2019parallel} showed that computing a Pareto optimal matching in the housing allocation model is in CC and computing the core of a housing market is CC-hard.

\subsection{Our Contributions}
\begin{enumerate}
    \item We give NC algorithms for both the popular matching problem and the maximum-cardinality popular matching problem in the setting of strictly-ordered preference lists.
    \item In the case that preference lists contain ties, we show that maximum-cardinality bipartite matching is NC-reducible to popular matching.
    \item We also give an NC algorithm to find the ``next" stable matching if one stable matching is given. We will define ``next" in Section \ref{sec:stable}. 
\end{enumerate}

\subsection{Organization of the paper}
In Section \ref{sec:pre}, we give the terminology and notation of popular matchings and stable matchings. In Section \ref{sec:pop}, we give an NC algorithm for the popular matching problem without ties. In Section \ref{sec:max-pop}, we give an NC algorithm for the maximum-cardinality popular matching problem. In Section \ref{sec:stable}, we give an NC algorithm for finding the ``next" stable matching. Finally, we give open problems related to our work.
\section{Preliminaries} \label{sec:pre}
\subsection{Popular Matching Problem}
Let $\mathcal{A}$ be a set of applicants and $\mathcal{P}$ be a set of posts, associated with each member of $\mathcal{A}$ is a preference list (possibly involving ties) comprising a non-empty subset of the elements of $\mathcal{P}$. An instance of the {\em popular matching problem} is a bipartite graph $G = (\mathcal{A} \cup \mathcal{P}, E)$ and a partition $E = E_1 \dot\cup E_2 \dot\cup \cdots \dot\cup E_r$ of the edge set. The partition $E$ consists of all pairs $(a,p)$ such that post $p$ appears in the preference list of applicant $a$ and we say that each edge $(a,p) \in E_i$ has a rank $i$ if post $p$ is on the $i$-th position of the preference list of applicant $a$. If $(a,p) \in E_i$ and $(a,p') \in E_j$ with $i < j$, we say that $a$ prefers $p$ to $p'$. If $i=j$, we say that $a$ is indifferent between $p$ and $p'$. We say that preference lists are strictly ordered if no applicant is indifferent between any two posts on his/her preference list. Otherwise, we say that preference lists contain ties.

A {\em matching} $M$ of $G$ is a set of edges no two of which share an endpoint. A node $u \in \mathcal{A} \cup \mathcal{P}$ is either unmatched or matched to some node, denoted by $M(u)$. We say that an applicant $a$ prefers matching $M'$ to $M$ if (i) $a$ is matched in $M'$ and unmatched in $M$, or (ii) $a$ is matched in both $M'$ and $M$, and $a$ prefers $M'(a)$ to $M(a)$. Let $\mathcal{M}$ be the set of matchings in $G$ and let $M, M' \in \mathcal{M}$. Let $P(M, M')$ denote the set of applicants who prefer $M$ to $M'$. Define a ``{\em more popular than}" relation $\succ$ on $M$ as follows: if $M,M' \in \mathcal{M}$, then $M'$ is {\em more popular than} $M$, denoted by $M' \succ M$, if $|P(M', M)| > |P(M, M')|$. 

\begin{definition}
A matching $M \in \mathcal{M}$ is popular if there is no matching $M'$ such that $M' \succ M$.
\end{definition}

The popular matching problem is to determine if a given instance admits a popular matching, and to find such a matching, if one exists. Note that popular matchings may have different sizes, a largest popular matching may be smaller than a maximum-cardinality matching since no maximum-cardinality matching needs to be popular. The maximum-cardinality popular matching problem then is to determine if a given instance admits a popular matching, and to find a largest such matching, if one exists. Figure \ref{fig:pop} shows an example of a popular matching instance. The reader can check that $\{(a_1,p_1),(a_2,p_2),(a_3,p_4),(a_4,p_3),(a_5,p_5),(a_6,p_7),$\\
$(a_7,p_8),(a_8,p_9)\}$ is a popular matching.

\begin{figure}[h]
    \centering
    $\begin{aligned}
        a_1: &\sh p_1\sh p_4\sh p_5\sh p_2\sh p_6\\
        a_2: &\sh p_4\sh p_5\sh p_7\sh p_2\sh p_8\\
        a_3: &\sh p_4\sh p_1\sh p_3\sh p_8\\
        a_4: &\sh p_1\sh p_7\sh p_4\sh p_3\sh p_9\\
        a_5: &\sh p_5\sh p_1\sh p_7\sh p_2\sh p_6\\
        a_6: &\sh p_7\sh p_6\\
        a_7: &\sh p_7\sh p_4\sh p_8\sh p_2\\
        a_8: &\sh p_7\sh p_4\sh p_1\sh p_5\sh p_9\sh p_3
    \end{aligned}$
    \caption{A popular matching instance $I$}
    \label{fig:pop}
\end{figure}

As in \cite{abraham2007popular}, we add a unique {\em last resort post} $l(a)$ for each applicant $a$ and assign the edge $(a, l(a))$ higher rank than any edge incident on $a$. In this way, we can assume that every applicant is matched, since any unmatched applicant can be matched to his/her unique last resort post. From now on, we only focus on matchings that are {\em applicant-complete}, and the size of a matching is the number of applicants not matched to their last resort posts.

\begin{definition}
A matching $M\in \mathcal{M}$ is applicant-complete if each applicant $a\in\mathcal{A}$ is matched to some post $p\in\mathcal{P}$. 
\end{definition}

\section{Finding Popular Matching in NC} \label{sec:pop}

\subsection{Characterizing Popular Matchings}
We restrict our attention to strictly-ordered preference lists. For each applicant $a$, let $f(a)$ denote the first-ranked post on $a$'s preference list. We call any such post $p$ an $f$-post, and denote by $f^{-1}(p)$ the set of applicants $a$ for which $f(a) = p$. For each applicant $a$, let $s(a)$ denote the first non-$f$-post on $a$'s preference list (note that $s(a)$ always exists, due to the introduction of $l(a)$). We call any such post $p$ an $s$-post, and remark that $f$-posts are disjoint from $s$-posts. We also call any last resort post $p$ an $l$-post.

The following theorem, proved in \cite{abraham2007popular}, completely characterizes popular matchings.

\begin{theorem} \label{thm:popular-matching}
A matching $M$ is popular if and only if 
\begin{enumerate}[(i)]
    \item every f-post is matched in $M$, and
    \item for each applicant $a$, $M(a) \in \{f(a), s(a)\}$.
\end{enumerate}
\end{theorem}

Let $G'$ be the reduced graph of $G$ that only includes $f$-posts and $s$-posts. For a reduced graph $G'$, let $M$ be a popular matching, and let $a$ be an applicant. Denote by $O_M(a)$ the post on $a$'s reduced preference list to which $a$ is not assigned in $M$. Note that since $G'$ is a reduced graph of $G$, $O_M(a)$ is well-defined. If $a$ is matched to $f(a)$ in $M$, then $O_M(a) = s(a)$, whereas if $a$ is matched to $s(a)$ in $M$, then $O_M(a) = f(a)$.

\subsection{Algorithmic Results}
Now we show Algorithm \ref{alg:popular-matching} is an NC algorithm for the popular matching problem with strictly-ordered preference lists. First we construct the reduced graph $G'$ from $G$. Then we find an applicant-complete matching $M$ in $G'$. Hence for each applicant $a$, $M(a) \in \{f(a),s(a)\}$. Then for any $f$-post $p$ that is unmatched in $M$, we match $p$ with any applicant in $f^{-1}(p)$.

The most non-trivial part is line 4 that determines an applicant-complete matching $M$ in the reduced graph $G'$. Perfect matching in bipartite graph is in Quasi-NC \cite{fenner2016bipartite}, but we do not know whether it is in NC. Recent results in \cite{anari2018planar} show that perfect matching in planar graph is in NC. But the reduced graph for popular matching problem is not necessarily planar. It is easy to check that the reduced graph $G'$ may contain a subgraph that is a subdivision of the complete bipartite graph $K_{3,3}$.

\begin{algorithm}
{\bf Input:} Graph $G = (\mathcal{A} \cup \mathcal{P} , E)$.\\
{\bf Output:} A popular matching $M$ or determine that no such matching.\\
$G' :=$ reduced graph of $G$;\\
{\bf if} $G'$ admits an applicant-complete matching $M$ {\bf then}\\
\h {\bf for each} $f$-post $p$ unmatched in $M$ {\bf in parallel do}\\
\h\h let $a$ be any applicant in $f^{-1}(p)$;\\
\h\h promote $a$ to $p$ in $M$;\\
\h {\bf return} $M$;\\
{\bf else}\\
\h {\bf return} ``no popular matching";\\
\caption{Popular Matching}
\label{alg:popular-matching}
\end{algorithm}

We first show how to construct the reduced graph $G'$ from $G$ in parallel (line 3). For each post $p$, we check if there is any incident edge $(a,p) \in E_1$. Let $\mathcal{F}$ be the set of such posts, which corresponds to all $f$-posts. Then for each post $p \in \mathcal{F}$, we remove all incident edges $(a,p) \notin E_1$. After that, for each applicant $a$, we find the highest ranked incident edge $(a,p) \notin E_1$, which corresponds to $s(a)$, and remove all other incident edges. The remaining graph must be $G'$. It is clear that each step can be done in logarithmic time with a polynomial number of operations.

It remains to show how to find an applicant-complete matching in $G'$ (line 4), or determine that no such matching exists in NC. Now we explain Algorithm \ref{alg:applicant-complete} that finds an applicant-complete matching.

\begin{algorithm}
{\bf Input:} Graph $G' = (\mathcal{A} \cup \mathcal{P}, E')$.\\
{\bf Output:} An applicant-complete matching $M$ or determine that no such matching exists.\\
$M := \emptyset$;\\
{\bf while} some post $p$ has degree $1$\\
\h Find all maximal paths that end at $p$;\\
\h {\bf for each} edge $(p',a')$ at an even distance from some $p$ {\bf in parallel do}\\
\h\h $M:=M \cup \{(p',a')\}$;\\
\h\h $G':=G'-\{p',a'\}$;\\
{\bf for each} post $p$ has degree $0$ {\bf in parallel do}\\
\h $G':=G'-{p}$\\
// Every post now has degree at least $2$;\\
// Every applicant still has degree $2$;\\
{\bf if} $|\mathcal{P}| < |\mathcal{A}|$ {\bf then}\\
\h {\bf return} ``no applicant-complete matching";\\
{\bf else}\\
\h // $G'$ decomposes into a family of disjoint even cycles\\
\h $M':=$ any perfect matching of $G'$;\\
\h {\bf return} $M \cup M'$;\\
\caption{Applicant-Complete Matching}
\label{alg:applicant-complete}
\end{algorithm}

The {\em while} loop (line 4) gradually matches applicants to posts of degree 1 or 2 until there is no post of degree 1. Then, either the remaining graph admits a perfect matching or we can conclude that there is no applicant-complete matching. We show the details below.

First, we identify all vertices of degree $2$ in $G'$. Note that all applicants have degree $2$, but posts may have any degree. We only need to identify posts of degree $2$. Some of these vertices might be connected to each other, in which case we get paths formed by these vertices. We can extend these paths, by the doubling trick in polylog time to find maximal paths consisting of degree 2 vertices. Let the vertices of the path be $(v_1, v_2, \cdots, v_k)$. Further, let $v_0$ be the vertex we would get if we extended this path from $v_1$ side and $v_{k+1}$ be the one we would get from $v_k$ side. Note that $\deg(v_i) = 2$ for $i=1,\cdots,k$ but not for $i=0,k+1$. 

Then, in parallel, we consider each maximal path with at least one of $v_0$ and $v_{k+1}$ of degree $1$. W.l.o.g, let $v_0$ be the vertex of degree $1$. For each such path, we add each edge at an even distance from $v_0$ to $M$ (e.g. the edge $(v_0, v_1)$ is at zero distance from $v_0$ and must be added to $M$) and delete $v_0,\cdots, v_k$ and their incident edges. Note that $v_0$ and $v_{k+1}$ can only be posts since all applicants have degree exactly $2$. Hence, any maximal path must have even length and $v_{k+1}$ is not matched. In the case both end points have degree $1$, we only consider this path once and choose $v_0$ or $v_{k+1}$ to be matched arbitrarily. After one round, there would be some new vertices of degree $1$ because the degree of $v_{k+1}$ decreases by 1 for each maximal path that ends at $v_{k+1}$. Run the same process until there is no post that has degree 1. After removing any isolated posts, we can conclude that either there is no applicant-complete matching, or the remaining graph is a family of disjoint even cycles.

\subsubsection{Correctness}
Algorithm \ref{alg:applicant-complete} begins by repeatedly matching maximal paths $(v_0, v_1, \cdots, v_{k+1})$ with $\deg(v_0) = 1$. After first round, no subsequent augmenting path can include any vertices $v_i$ for $i = 0,1,\cdots, k$ since they are matched and any alternating path that includes them must end at $v_0$, which is matched and has degree $1$. So we can remove all matched vertices from consideration. The same argument holds for subsequent rounds. Also note that the while loop always terminates because whenever we find a post of degree $1$, we match at least one edge $(v_0,v_1)$ and remove at least two vertices that are $\{v_0, v_1\}$. 

Now we have a matching and we only need to match remaining posts and applicants. All remaining posts have degree at least $2$, while all remaining applicants still have degree exactly $2$. Now, if $|\mathcal{P}| < |\mathcal{A}|$, $G'$ cannot admit an applicant-complete matching by Hall's Marriage Theorem \cite{hall2009representatives}. Otherwise, we have that $|\mathcal{P}| \geq |\mathcal{A}|$, and by a double counting argument, we have $2|\mathcal{P}| \leq \sum_{p \in \mathcal{P}} \deg(p) = 2|\mathcal{A}|$. Hence, it must be that $|\mathcal{P}| = |\mathcal{A}|$ and every post has degree exactly $2$. $G'$ becomes $2$-regular bipartite graph and consists of disjoint union of even cycles. By choosing any edge $e$ in an even-length cycle $C$, even distance (resp. odd distance) from $e$ is well-defined. Choosing all edges of even distance yields a perfect matching in $G'$. Now we have an applicant-complete matching in $G'$. Hence for each applicant $a$, $M(a) \in \{f(a),s(a)\}$. Then for any $f$-post $p$ that is unmatched in $M$, we match $p$ with any applicant in $f^{-1}(p)$. By Theorem \ref{thm:popular-matching}, the resulting matching is a popular matching.

\subsubsection{Complexity}
Lemma \ref{lem:while} proves that the {\em while} loop in Algorithm \ref{alg:applicant-complete} runs $O(\log(n))$ number of times.

\begin{lemma} \label{lem:while}
The \em{while} loop (line 4) runs $O(\log(n))$ number of times.
\end{lemma}
\begin{proof}

For any vertex $v$ of $\deg(v) \geq 3$ that is reduced to degree of $1$, it must be the end point of $\deg(v)-1$ maximal paths. If in round $r$, s.t. $r > 1$, there are $t$ vertices of degree $1$ ( for some constant $t$), then we must have deleted at least $2t$ vertices in round $r-1$. After round $r$, we have deleted at least $(2^r-1)t$ vertices. Hence, it is clear that the while loop can be run at most $\lceil\log(n)\rceil + 1$ times since the total number of vertices is bounded by $n$. 
\end{proof}

Finding all maximal paths of degree $2$ vertices and calculating the distance from $v_0$ in the path can be done in polylog time. Furthermore, the {\em while} loop runs at most a logarithmic number of times. Finding a prefect matching in a $2$-regular bipartite graph i.e. graph consisting of even-length cycles is in NC. More generally, searching for a perfect matching in regular bipartite graphs can be done in NC \cite{lev1981fast}. So, Algorithm \ref{alg:applicant-complete} is in NC. The {\em for} loop in Algorithm \ref{alg:popular-matching} can be done in constant time since for every $f$-post $p$, $f^{-1}(p)$ is disjoint from each other. 

We summarize the preceding discussion in the following theorem.
\begin{theorem}
We can find a popular matching, or determine that no such matching exists in NC.
\end{theorem}

\subsection{Example of Popular Matchings}\label{app:example}
To illustrate Algorithm \ref{alg:applicant-complete}, we provide a detailed example. Figure \ref{fig:pop} shows the preference lists for a popular matching instance $I$.  The set of $f$-posts is $\{p_1, p_4, p_5, p_7\}$ and the set of $s$-posts is $\{p_2, p_3, p_6, p_8, p_9\}$. 

Figure \ref{fig:reduced-pop} shows the reduced preference lists of $I$ and reduced graph $G'$.
\begin{figure}[ht]
\centering
    \begin{subfigure}[b]{0.45\textwidth}
        \centering
        $\begin{aligned}
            a_1: &\sh \underline{p_1}\sh p_2\\
            a_2: &\sh p_4\sh \underline{p_2}\\
            a_3: &\sh \underline{p_4}\sh p_3\\
            a_4: &\sh p_1\sh \underline{p_3}\\
            a_5: &\sh \underline{p_5}\sh p_2\\
            a_6: &\sh \underline{p_7}\sh p_6\\
            a_7: &\sh p_7\sh \underline{p_8}\\
            a_8: &\sh p_7\sh \underline{p_9}
        \end{aligned}$
        \caption{The reduced preference lists of $I$ with popular matching $M$ denoted by underlining}
    \end{subfigure}
    ~
    \begin{subfigure}[b]{0.45\textwidth}
        \centering
        \begin{tikzpicture}
        \draw[fill=black] (0,0.5) circle (2pt);
        \draw[fill=black] (0,1) circle (2pt);
        \draw[fill=black] (0,1.5) circle (2pt);
        \draw[fill=black] (0,2) circle (2pt);
        \draw[fill=black] (0,2.5) circle (2pt);
        \draw[fill=black] (0,3) circle (2pt);
        \draw[fill=black] (0,3.5) circle (2pt);
        \draw[fill=black] (0,4) circle (2pt);
        \draw[fill=black] (1.5,0) circle (2pt);
        \draw[fill=black] (1.5,0.5) circle (2pt);
        \draw[fill=black] (1.5,1) circle (2pt);
        \draw[fill=black] (1.5,1.5) circle (2pt);
        \draw[fill=black] (1.5,2) circle (2pt);
        \draw[fill=black] (1.5,2.5) circle (2pt);
        \draw[fill=black] (1.5,3) circle (2pt);
        \draw[fill=black] (1.5,3.5) circle (2pt);
        \draw[fill=black] (1.5,4) circle (2pt);
        \node at (-0.5,0.5) {$a_8$};
        \node at (-0.5,1) {$a_7$};
        \node at (-0.5,1.5) {$a_6$};
        \node at (-0.5,2) {$a_5$};
        \node at (-0.5,2.5) {$a_4$};
        \node at (-0.5,3) {$a_3$};
        \node at (-0.5,3.5) {$a_2$};
        \node at (-0.5,4) {$a_1$};
        \node at (2,0) {$p_9$};
        \node at (2,0.5) {$p_8$};
        \node at (2,1) {$p_7$};
        \node at (2,1.5) {$p_6$};
        \node at (2,2) {$p_5$};
        \node at (2,2.5) {$p_4$};
        \node at (2,3) {$p_3$};
        \node at (2,3.5) {$p_2$};
        \node at (2,4) {$p_1$};
        \draw[thick] (1.5,0) -- (0,0.5) -- (1.5,1) -- (0,1.5) -- (1.5,1.5) (1.5,0.5) -- (0,1) -- (1.5,1) (1.5,2) -- (0,2) -- (1.5,3.5) -- (0,3.5) -- (1.5,2.5) -- (0,3) -- (1.5,3) -- (0,2.5) -- (1.5,4) -- (0,4) -- (1.5,3.5);
        \end{tikzpicture}
        \caption{The reduced graph $G'$ of $G$}
    \end{subfigure}

    \caption{The reduced instance of $I$}
    \label{fig:reduced-pop}
\end{figure}
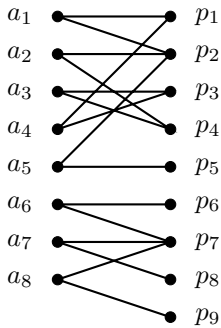

In the while loop of Algorithm \ref{alg:applicant-complete}, pairs $(a_8,p_9),(a_6,p_6)$ $,(a_7,p_8),(a_5,p_5)$ are matched. Figure \ref{fig:cycles} shows the reduced graph after the while loop of Algorithm \ref{alg:applicant-complete}. The graph consists of only even-length cycles. Choose one perfect matching in the reduced graph such as pairs $(a_1,p_1),(a_2,p_2),(a_3,p_4),(a_4,p_3)$, we obtain an applicant-complete matching. Note that one $f$-post $p_7$ is not matched in this applicant-complete matching. So we can promote any applicant from $\{a_6,a_7,a_8\}$ to match with $p_7$, e.g. $a_6$ is matched to $p_7$. The resulting popular matching $M$ is $\{(a_1,p_1),(a_2,p_2),(a_3,p_4),(a_4,p_3),(a_5,p_5),(a_6,p_7),(a_7,p_8)$\\$,(a_8,p_9)\}$.

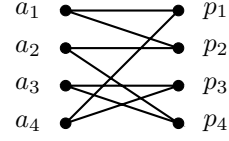
\begin{figure}[h]
    \centering
        \begin{tikzpicture}
        \draw[fill=black] (0,0) circle (2pt);
        \draw[fill=black] (0,0.5) circle (2pt);
        \draw[fill=black] (0,1) circle (2pt);
        \draw[fill=black] (0,1.5) circle (2pt);
        \draw[fill=black] (1.5,0) circle (2pt);
        \draw[fill=black] (1.5,0.5) circle (2pt);
        \draw[fill=black] (1.5,1) circle (2pt);
        \draw[fill=black] (1.5,1.5) circle (2pt);
        \node at (-0.5,0) {$a_4$};
        \node at (-0.5,0.5) {$a_3$};
        \node at (-0.5,1) {$a_2$};
        \node at (-0.5,1.5) {$a_1$};
        \node at (2,0) {$p_4$};
        \node at (2,0.5) {$p_3$};
        \node at (2,1) {$p_2$};
        \node at (2,1.5) {$p_1$};
        \draw[thick] (1.5,1) -- (0,1) -- (1.5,0) -- (0,0.5) -- (1.5,0.5) -- (0,0) -- (1.5,1.5) -- (0,1.5) -- (1.5,1);
        \end{tikzpicture}
    \caption{The reduced graph after the while loop of Algorithm \ref{alg:applicant-complete}}
    \label{fig:cycles}
\end{figure}
\section{Finding Maximum-Cardinality Popular Matching in NC} \label{sec:max-pop}
We now consider the maximum-cardinality popular matching problem. Let $\mathcal{A}_1$ be the set of all applicants $a$ with $s(a) = l(a)$, and let $\mathcal{A}_2 = \mathcal{A} - \mathcal{A}_1$. Our target matching must satisfy conditions (i) and (ii) of Theorem \ref{thm:popular-matching}, and among all such matchings, allocate the fewest $\mathcal{A}_1$-applicants to their last resort. 
To be able to find maximum-cardinality matching in NC, we need another characterization of popular matching problem called {\em switching graph} \cite{mcdermid2011popular}, a directed graph which captures all the possible ways in which applicants may form different popular matchings by switching between the two posts on their reduced preference lists. 

Given a popular matching $M$ for an instance $G = (\mathcal{A} \cup \mathcal{P}, E)$, the switching graph $G_M$ of $M$ is a directed graph with a vertex for each post $p$, and a directed edge $(p_i, p_j)$ for each applicant $a$, where $p_i = M(a)$ and $p_j = O_M(a)$. Then each edge is labelled with the applicant that it represents. A {\em component} of $G_M$ is any maximal weakly connected subgraph of $G_M$. An applicant (resp. post) is said to be in a component, or path, or cycle of $G_M$ if the edge (resp. vertex) representing it is in that component, path or cycle. The following lemma in \cite{mcdermid2011popular} gives some simple properties of switching graphs.

\begin{lemma}[\cite{mcdermid2011popular}, Lemma 1] \label{lem:switching-graph}
Let $M$ be a popular matching for an instance of $G = (\mathcal{A} \cup \mathcal{P}, E)$, $G_M$ be the switching graph of $M$. Then
\begin{enumerate}[(i)]
    \item Each vertex in $G_M$ has outdegree at most $1$.
    \item The sink vertices of $G_M$ are those vertices corresponding to posts that are unmatched in $M$, and are all $s$-post vertices.
    \item Each component of $G_M$ contains either a single sink vertex or a single cycle.
\end{enumerate}
\end{lemma}

A component of a switching graph $G_M$ is called a {\em cycle component} if it contains a cycle, and a {\em tree component} if it contains a sink vertex. Each cycle in $G_M$ is called a {\em switching cycle}. If $T$ is a tree component of $G_M$ with sink vertex $p$, and if $q$ is another $s$-post vertex in $T$, the unique path from $q$ to $p$ is called a {\em switching path}. Note that each cycle component of $G_M$ has a unique switching cycle, but each tree component may have zero or multiple switching paths; to be precise it has one switching path for each $s$-post vertex it contains, other than the sink vertex. 

Figure \ref{fig:switching-graph} shows the switching graph $G_M$ for popular matching $M$. There are one switching cycle and two switching paths starting from $p_8$ and $p_9$ respectfully.

\begin{figure}[h]
    \centering
    \begin{tikzpicture}[> = stealth, auto, node distance=3cm, ultra thick,
   node_style/.style={circle,draw=black,font=\sffamily\Large\bfseries},
   edge_style/.style={draw=black, ultra thick}]
    \node[node_style] (p1) at (-4,2) {$p_1$};
    \node[node_style] (p2) at (-2,2) {$p_2$};
    \node[node_style] (p3) at (-4,0) {$p_3$};
    \node[node_style] (p4) at (-2,0) {$p_4$};
    \node[node_style] (p5) at (0,3) {$p_5$};
    \node[node_style] (p6) at (0,0) {$p_6$};
    \node[node_style] (p7) at (1,2) {$p_7$};
    \node[node_style] (p8) at (3,2) {$p_8$};
    \node[node_style] (p9) at (3,0) {$p_9$};
    \path[->] (p1) edge node {$a_1$} (p2);
    \path[->] (p2) edge node {$a_2$} (p4);
    \path[->] (p4) edge node {$a_3$} (p3);
    \path[->] (p3) edge node {$a_4$} (p1);
    \path[->] (p5) edge node {$a_5$} (p2);
    \path[->] (p7) edge node {$a_6$} (p6);
    \path[->] (p8) edge node {$a_7$} (p7);
    \path[->] (p9) edge node {$a_8$} (p7);
    \end{tikzpicture}
    \caption{The switching graph $G_M$ for popular matching $M$}
    \label{fig:switching-graph}
\end{figure}
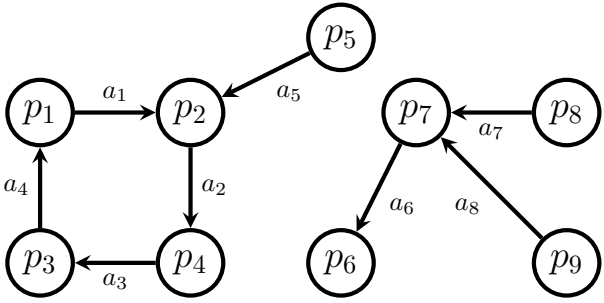

Lemma \ref{lem:switching-graph} shows that the switching graph $G_M$ is indeed a directed pseudoforest. Next we give several NC algorithms for finding all switching cycles and switching paths in $G_M$. 

\subsection{Finding Cycles in Pseudoforest in NC} \label{sec:pseudoforest}
\begin{definition}
A {\bf pseudoforest} is an undirected graph in which every connected component has at most one cycle. A {\bf pseudotree} is a connected pseudoforest. A {\bf directed pseudoforest} is a directed graph in which each vertex has at most one outgoing edge, i.e., it has outdegree at most one. A {\bf directed 1-forest} (most commonly called a functional graph, sometimes maximal directed pseudoforest) is a directed graph in which each vertex has outdegree exactly one.
\end{definition}

It is easy to see that every weakly connected component in a directed pseudoforest contains either a single sink vertex or a single cycle. 

We consider the problem of finding switching cycles in $G_M$, later we will show that finding switching paths is as easy as finding switching cycles.

Given a directed pseudoforest $G_P$, we want to find each unique cycle $C$ in each component of $G_P$. There could not be any cycle in a component of $G_M$ if it is a tree component. The first approach is based on transitive closure $G_P^*$ of $G_P$ since computing the transitive closure is in NC by Theorem \ref{thm:transitive-closure}. We compute the transitive closure $G_P^*$ and for any two vertices $i$ and $j$ s.t. $i \neq j$ in $G_P$, if $G_P^*(i,j) = 1$ and $G_P^*(j,i) = 1$, then both $i$ and $j$ are in the unique cycle $C$. Hence we can identify the cycle $C$ by checking each pair of vertices in parallel. 

\begin{theorem}[\cite{jaja1992introduction}] \label{thm:transitive-closure}
The transitive closure of a directed graph with $n$ vertices can be computed in $O(\log^2 n)$ time, using $O(M(n)\log n)$ operations on a CREW PRAM, where $M(n)$ is the best known sequential bound for multiplying two $n \times n$ matrices over a ring.
\end{theorem}                        

We also give NC algorithms in the setting of undirected graph in which transitive closure does not help. Given an undirected pseudoforest $G_P$, denote the incidence matrix of $G_P$ as $I_{G_P}$. Let $cc(G)$ be the number of connected components in $G$. The basic idea is that we remove any one edge $e$ from $G_P$, if $e \in C$ s.t. $C$ is the unique cycle in $G_P$, then $cc(G_P - \{e\}) = cc(G_P)$; otherwise, $cc(G_P - \{e\}) = cc(G_P) + 1$. There is a direct connection between the rank of incidence matrix $I$ of $G$ and the number of connected component $cc(G)$ in $G$.

\begin{lemma}
If $G$ is an undirected graph with $k$ connected components, then the rank of its incidence matrix $I_G$ is $n-k$.
\end{lemma}

So we can compute the rank of $I_{G_P}$ and for each $e$ in $G_P$, compute the rank of $I_{G_P - \{e\}}$ in parallel. There are at most $|V|$ edges in $G_P$. 

\begin{theorem}[\cite{mulmuley1987fast}] \label{thm:rank-nc}
The rank of a $n \times n$ matrix over an arbitrary field can be computed in $O(\log^2 n)$ time, using a polynomial number of processors.
\end{theorem}

We can also compute the number of connected component of $G_P$ directly by finding all connected components in $G_P$.

\begin{theorem}[\cite{cole1986approximate}] \label{thm:cc-nc}
The connected components of a graph with $n$ vertices and $m$ edges can be computed in $O(\log n)$ time, using $O((m+n)\alpha(m,n)/\log n)$ operations on an ARBITRARY CRCW PRAM, where $\alpha(m,n)$ is the inverse Ackermann function.
\end{theorem}

For any tree component $T$, there might be zero or multiple switching paths. For each $s$-post $p$, we make a copy of $T$ and add one directed edge from the sink vertex to $p$ and then find the unique cycle in the new graph, which yields one switching path in $T$. 

\subsection{Algorithmic Results}
Now we are ready to give an NC algorithm to find a maximum-cardinality popular matching. 
\begin{algorithm}
{\bf Input:} Reduced graph $G' = (\mathcal{A} \cup \mathcal{P}, E')$ and a popular matching $M$.\\
{\bf Output:} A maximum-cardinality popular matching $M'$.\\
$G_M :=$ switching graph of $M$ and $G'$.\\
Find all weakly connected components of $G_M$;\\
{\bf for each} cycle component (resp. tree component) {\bf in parallel do}\\
\h Find the unique switching cycle (resp. each switching path);\\
{\bf for each} switching cycle (resp. switching path) {\bf in parallel do}\\
\h Compute the margin of applying this switching cycle(resp. switching path);\\
{\bf for each} cycle component (resp. tree component) {\bf in parallel do}\\
\h {\bf if} the margin $\Delta$ of switching cycle (resp. the largest margin of switching paths) is positive\\
\h\h Apply this switching cycle (resp. switching path) to $M$;\\
// The resulting matching $M'$ after applying such switching cycles and switching paths is the maximum-cardinality matching.\\
{\bf return} $M'$;\\
\caption{Maximum-Cardinality Popular Matching}
\label{alg:max-cardinality}
\end{algorithm}

Given the reduced graph $G'$ and a popular matching $M$, we construct the switching graph $G_M$. After that, we identify the unique switching cycle or each switching path in $G_M$. Then we increase the size of popular matching locally according to the margin $\Delta$ of each component.

For each switching cycle $C$ (resp. switching path $P$), we define the margin $\Delta$ in Definition \ref{def:margin} as the difference of the number of {\em last resort posts} after applying $C$ (resp. $P$) to $M$. For each applicant $a$, the margin $\Delta$ increases by $1$ if $a$ promotes from $l(a)$ to $f(a)$ or decreases by $1$ if $a$ demotes from $f(a)$ to $l(a)$, otherwise no change is made. The following theorem gives a one-to-one correspondence between a popular matching and a unique subset of the cycle components and the tree components of $G_M$, which is crucial to our algorithm for maximum-cardinality popular matching.

\begin{definition}\label{def:margin}
Let $\Delta$ be the margin of applying a switching cycle $C$ (resp. switching path $P$) to $M$, i.e.\\
\centerline{$\Delta = \sum_{a \in C (\text{resp.} P)} \mathbbm{1}_{M\cdot C(a)} - \mathbbm{1}_{M(a)}$}
where $\mathbbm{1}_p$ is an indicator function of posts \\
s.t. $\mathbbm{1}_p := \left\{ 
\begin{aligned}
1 & ~~ \text{if}~ p ~\text{is not}~ l \text{-post}  \\
0 & ~~ \text{if}~ p ~\text{is}~ l\text{-post}
\end{aligned}\right.
$
\end{definition}

The following theorem is crucial for the correctness of Algorithm \ref{alg:max-cardinality}.
\begin{theorem}[\cite{mcdermid2011popular}, Corollary 1] \label{thm:one-to-one}
Let $G = (\mathcal{A} \cup \mathcal{P}, E)$ be an instance, and let $M$ be an arbitrary popular matching for $G$ with switching graph $G_M$. Let the tree components of $G_M$ be $T_1, T_2, \cdots, T_k$, and the cycle components of $G_M$ be $C_1, C_2, \cdots, C_l$. Then the set of popular matchings for $G$ consists of exactly those matchings obtained by applying at most one switching path in $T_i$ for each $i (1\leq i \leq k)$ and by either applying or not applying the switching cycle in $C_i$ for each $i (1\leq i \leq l)$.
\end{theorem}

\subsection{Correctness}
 Any popular matching can be obtained from $M$ by applying at most one switching cycle or switching path per component of the switching graph $G_M$. For any tree component $T$, we apply the switching path in $T$ with the largest positive margin. Similarly, for any cycle component $C$, we apply the switching cycle in $C$ with positive margin. Then, we get the largest possible total margin, which in turn implies the largest possible number of $l$-posts we removed from $M$. Hence, we obtain the maximum-cardinality popular matching. For any other popular matching obtained by applying difference subset of switching paths or switching cycles, it will have strictly less total margin than the maximum-cardinality popular matching. 

\subsection{Complexity}
It is clear that the switching graph $G_M$ can be constructed from $G'$ and $M$ in constant time in parallel. All weakly connected components of $G_M$ can also be found in polylog time by Theorem \ref{thm:cc-nc}. Moreover, in Section \ref{sec:pseudoforest}, we showed that all switching cycles and switching paths can be found in polylog time. Each switching cycle and switching path can be applied to matching $M$ easily in parallel since they are vertex-disjoint in $G_M$. So, overall the complexity of Algorithm \ref{alg:max-cardinality} is $O(\log^2 n)$.

We summarize the preceding discussion in the following theorem.
\begin{theorem}
We can find a maximum-cardinality popular matching, or determine that no such matching exists in NC.
\end{theorem}

\subsection{Optimal Popular Matchings}
It is natural to extend the popular matching problem to a weighted version of the popular matching problem. If a weight $w(a_i,p_j)$ is defined for each applicant-post pair with $p_j$ acceptable to $a_i$, then the weight $w(M)$ of a popular matching $M$ is $\sum_{(a_i,p_j)\in M} w(a_i, p_j)$. A popular matching is optimal if it is a maximum or minimum weight popular matching. It turns out that maximum-cardinality popular matching is a special case of maximum weight popular matching if we assign a weight of 0 to each pair involving a last resort post and a weight of 1 to all other pairs. 

Kavitha et al. \cite{kavitha2009optimal} considered other optimality criteria, in terms of the so called profile of the matching. For a popular matching instance with $n_1$ applicants and $n_2$ posts, we define the profile $\rho(M)$ of $M$ to be the $(n_2+1)$ tuple $(x_1, x_2, \cdots, x_{n_2+1})$ such that for each $i, 1\leq i \leq n_2+1$, $x_i$ is the number of applicants who are matched with their $i$th ranked post. An applicant who is matched to his/her last resort post is considered to be matched to his/her $(n_2+1)$th ranked post, regardless of the length of his/her preference list. 

Suppose that $\rho = (x_1, x_2, \cdots, x_{n_2+1})$ and $\rho' = (y_1, y_2, \cdots$ $ , y_{n_2+1})$. We use $\succ_R$ denote the lexicographic order on profiles: $\rho \succ_R \rho'$ if $x_i = y_i$ for $1 \leq i < k$ and $x_k > y_k$, for some $k$. Similarly, we use $\prec_F$ to denote the lexicographic order on profiles: $\rho \prec_F \rho'$ if $x_i = y_i$ for $k < i \leq n_2+1$ and $x_k < y_k$, for some $k$.

A rank-maximal popular matching is a popular matching whose profile is maximal with respect to $\succ_R$. A fair popular matching is a popular matching whose profile is minimal with respect to $\prec_F$. Note that a fair popular matching is always a maximum-cardinality popular matching since the number of last resort posts is minimized. It is easy to check these two problems are equivalent to the optimal popular matching problem with suitable weight assignments as follows.
\begin{itemize}
    \item Rank-maximal popular matching: assign a weight of 0 to each pair involving a last resort post and a weight of $n_1^{n_2-k+1}$ to each pair $(a_i, p_j)$ where $p_j$ is $k$th ranked post of $a_i$, and find a maximum weight popular matching.
    
    \item Fair popular matching: assign a weight of $n_1^{k}$ to each pair $(a_i, p_j)$ where $p_j$ is the $k$th ranked post of $a_i$, and find a minimum weight popular matching.
\end{itemize}

Now we are ready to give an NC algorithm for the optimal popular matching problem. Given a popular matching instance and a particular weight assignment, let $M$ be a popular matching, and $M_{opt}$ be an optimal popular matching (maximum or minimum weight, depends on the context). By Theorem \ref{thm:one-to-one}, $M_{opt}$ can be obtained from $M$ by applying a choice of at most one switching cycle or switching path per component of the switching graph $G_M$. Similar to Algorithm \ref{alg:max-cardinality}, the algorithm for computing $M_{opt}$ will compute an arbitrary popular matching $M$, and make an appropriate choice of switching cycles and switching paths to apply in order to obtain an optimal popular matching. The only difference is the margin calculation. In order to decide to apply a switching cycle $C$ or not, we need to compare $\sum_{a\in C} w(a, M(a))$ with $\sum_{a \in C} w(a, M\cdot C(a))$. In the case of maximum-cardinality popular matching, the weight assignment is either 0 or 1. While in rank-maximal popular matching and fair popular matching, $w$ is bounded by $n_1^{n_2+1}$, which has $\tilde{O}(n)$ bits. So $\sum_{a\in C} w(a, M(a))$ can be computed in NC.

\section{Preference Lists with Ties}
In this section, we consider the popular matching problem such that preference lists are not strictly ordered, but contain ties. Without the assumption of strictly ordered preference lists, we show that the popular matching problem is at least as hard as the maximum-cardinality bipartite matching problem by showing that maximum-cardinality bipartite matching is NC-reducible to popular matching. Note that whether bipartite perfect matching is in NC is still open \cite{fenner2016bipartite}. 

Now we show the following NC reduction.
\begin{theorem} \label{thm:nc-reduction}
Maximum-cardinality Bipartite Matching $\leq_{NC}$ Popular Matching.
\end{theorem}

\begin{proof}
Suppose we have access to a black box that can solve Popular Matching in NC. Consider an arbitrary instance of Maximum-cardinality Bipartite Matching, specified by a graph $G = (\mathcal{A}\cup\mathcal{B}, E)$. We construct our Popular Matching instance by giving all edges rank 1, i.e. each applicant has the same preference over all acceptable posts. For convenience, we also use $G = (\mathcal{A}\cup\mathcal{B}, E)$ as our instance of Popular Matching. We do not add last resort posts at all. Lemma \ref{lem:right-reduction} and Lemma \ref{lem:left-reduction} show that popular matching always exists in $G$ and any popular matching $M$ is also a maximum-cardinality matching in $G$.
\end{proof}

\begin{lemma} \label{lem:right-reduction}
Let $M$ be a popular matching in $G$. Then $M$ is also a maximum-cardinality matching in $G$.
\end{lemma}
\begin{proof}
Suppose for a contradiction that $M$ is not a maximum matching of $G$. Then $M$ admits an augmenting path $Q$ with respect to $G$. Since each edge in $G$ has rank 1, after applying augmenting path $Q$ to $M$, we obtain a matching $M'$ that is {\em more popular than} $M$ because $M'$ has exactly one more edge matched than $M$ and all rest of applicants do not have preference over $M$ and $M'$.
\end{proof}

We know from Section \ref{sec:pop} that popular matching may not exist in an arbitrary popular matching instance. We show that given the construction that each edge in $G$ has rank 1, popular matching always exists.
\begin{lemma} \label{lem:left-reduction}
Let $M$ be a maximum-cardinality matching in $G$. Then $M$ is also a popular matching in $G$.
\end{lemma}

\begin{proof}
Consider any other matching $M'$ in $G$. We only care about the symmetric difference $M\Delta M'$ since the rest of edges do not have preference over $M$ and $M'$. Since all edges have rank 1, then $|P(M', M)| - |P(M, M')| = |M'| - |M| \leq 0$. Hence, no matching is {\em more popular than} $M$.
\end{proof}

We conjecture that the following reduction is also true.
\begin{conjecture} \label{conj:nc-reduction-left}
Popular Matching $\leq_{NC}$ Maximum-cardinality Bipartite Matching.
\end{conjecture}

\section{Finding ``next" Stable Matching in NC} \label{sec:stable}
In this section, we consider the problem of finding ``next" stable matching. \cite{gusfield1989stable} mentioned that even if it is not possible to find the first stable matching fast in parallel, perhaps, after sufficient preprocessing, the stable matchings could be enumerated in parallel, with small parallel time per matching. Our results partially answer this question, given a stable matching, we can enumerate the ``next" stable matching in the stable matching lattice in polylog time. This result can be regarded as an application of the techniques used in \ref{sec:max-pop}, that is to find cycles in pseudoforest in NC. The main result is given by Theorem \ref{thm:next-stable}.

We give some useful definitions in the next section.
\subsection{Stable Marriage Problem}
Let $\mathcal{A}$ be a set of $n$ men and $\mathcal{B}$ be a set of $n$ women. For any man $m\in \mathcal{A}$, there is a strictly ordered preference list containing all the women in $\mathcal{B}$. For any woman $w\in\mathcal{B}$, there is a strictly ordered preference list containing all the men in $\mathcal{A}$. Person $p$ prefers $q$ to $r$, where $q$ and $r$ are of the opposite sex to $p$, if and only if $q$ precedes $r$ on $p$'s preference list.

A {\em matching} $M$ is one-to-one correspondence between the men and the women. If man $m$ and woman $w$ are matched in $M$, then $m$ and $w$ are called partners in $M$, written as $m = p_M(w)$ and $w = p_M(m)$. A pair $(m,w)$ is called a {\em blocking pair} for $M$, if $m$ and $w$ are not partners in $M$, but $m$ prefers $w$ to $p_M(m)$ and $w$ prefers $m$ to $p_M(w)$.

\begin{definition}
A matching $M$ is stable if and only if there is no blocking pair for $M$.
\end{definition}

\begin{definition}[Partial Order $\mathcal{M}$]
For a given stable marriage instance, stable matching $M$ is said to dominate stable matching $M'$, written $M \preceq M'$, if every man either prefers $M$ to $M'$ or is indifferent between them. We use the term strictly dominate, written $M \prec M'$, if $M \preceq M'$ and $M \neq M'$. We use the symbol $\mathcal{M}$ to represent the set of all stable matchings for a stable marriage instance. Then the set $\mathcal{M}$ is a partial order under the dominance relation, denoted by $(\mathcal{M}, \preceq)$.
\end{definition}

It is well-known that the partial order $(\mathcal{M}, \preceq)$ forms a distributive lattice. Hence, the unique minimal element in $\mathcal{M}$ with respect to $\preceq$, i.e. man-optimal stable matching (denoted by $M_0$), as well the unique maximal element, i.e. woman-optimal stable matching (denoted by $M_z$) is well-defined. 

\begin{definition}[Rotation]\label{def:rotation}
Let $k \geq 2$. A rotation $\rho$ is an ordered list of pairs\\
\centerline{$\rho = ((m_0, w_0),(m_1, w_1), \cdots, (m_{k-1}, w_{k-1}))$}
that are matched in some stable matching $M$ with the property that for every $i$ such that $0 \leq i \leq k-1$, woman $w_{i+1}$ (where $i+1$ is taken modulo $k$) is the highest ranked woman on $m_i$'s preference list satisfying:
\begin{enumerate}[(i)]
    \item man $m_i$ prefers $w_i$ to $w_{i+1}$, and
    \item woman $w_{i+1}$ prefers $m_i$ to $m_{i+1}$.
\end{enumerate}
In this case, we say $\rho$ is exposed in $M$.
\end{definition}

\begin{definition}[Elimination of a Rotation]
Let $\rho = ((m_0, w_0)$\\$ ,(m_1, w_1), \cdots, (m_{k-1}, w_{k-1}))$ be a rotation exposed in a stable matching $M$. The rotation $\rho$ is eliminated from $M$ by matching $m_i$ to $w_{(i+1)\mod{k}}$ for all $0\leq i \leq k-1$, leaving all other pairs in $M$ unchanged, i.e. matching $M$ is replaced by matching $M'$, where\\
\centerline{$M' := M\backslash\rho\ \cup\ \{(m_0, w_1),(m_1, w_2), \cdots, (m_{k-1}, w_{0})\}$.}
Note that the resulting matching $M'$ is also stable.
\end{definition}

\begin{lemma}[\cite{gusfield1989stable}, Theorem 2.5.1] \label{immediate-ralation}
If $\rho$ is exposed in $M$, then $M$ immediately dominates $M\backslash\rho$, i.e. there is no stable matching $M'$ such that $M \prec M'$ and $M' \prec M\backslash\rho$.
\end{lemma}

\begin{theorem} \label{thm:next-stable}
Given a stable matching $M$, there is an NC algorithm that outputs stable matching $M\backslash\rho$ for each rotation $\rho$ exposed in $M$ or determines $M$ is the woman-optimal matching.
\end{theorem}

\subsection{Algorithmic Results}
We describe the NC algorithm to find the ``next" stable matching in this section. 

\begin{algorithm}
{\bf Input:} Stable matching $M$ and preference lists $mp$ and $wp$.\\
{\bf Output:} $M\backslash\rho$ or determine $M$ is the woman-optimal matching.\\
Compute ranking matrices $mr$ and $wr$; // constant steps\\
Compute reduced preference lists ${mp}'$ and ${wp}'$; // logarithmic number of steps\\
Construct $H_M$ from ${mp}'$;\\
{\bf if} $H_M$ is not empty {\bf then}\\
\h Find all simple cycles(rotations) in $H_M$;\\
\h {\bf for each} rotation $\rho$ in $H_M$ {\bf in parallel do}\\
\h\h return $M\backslash\rho$;\\
{\bf else}\\
\h return $M$ is the woman-optimal matching;\\
\caption{``next" Stable Matching}
\label{alg:stable-matching}
\end{algorithm}

Let $M$ be a stable matching. For any man $m$, let $s_M(m)$ denote the highest ranked woman on $m$'s preference list such that $w$ prefers $m$ to $p_M(w)$. Let $next_M(m)$ denote woman $s_M(m)$'s partner in $M$. Note that since $M$ is stable, $m$ prefers $p_M(m)$ to $s_M(m)$.

Now let $m$ be any man who has different partners in $M$ and $M_z$ and let $w$ be $m$'s partner in $M_z$. Since $M_z$ is woman-optimal, $m$ prefers $p_M(m)$ to $w$ and $w$ prefers $m$ to $p_M(w)$. Hence, $s_M(m)$ exists. If $s_M(m)$ exists and $m' = next_M(m)$, then $s_M(m')$ exists as well. Otherwise, $m'$ and $s_M(m)$ are partners in $M_z$, so $m$ prefers $s_M(m)$ to his partner $w$ in $M_z$ and $s_M(m)$ prefers $m$ to her partner $m'$ in $M_z$, contradicting the stability of $M_z$. Denote $D$ the set of man who has different partners in $M$ and $M_z$, then for any man $m \in D$, $next_M(m) \in D$. Later we will show that the algorithm does need to know $M_z$.

Similar to the {\em switching graph} of popular matching, we define the {\em switching graph} of stable matching $M$ as a directed graph $H_M$ with a vertex for each man in $D$ and a directed edge from the vertex for $m$ to the vertex for $next_M(m)$, which is also in $H_M$. Some simple properties of {\em switching graph} $H_M$ is shown in the following lemma.

\begin{lemma}
Let $M$ be a stable matching other than the woman-optimal matching $M_z$, let $H_M$ be the switching graph of $M$, then
\begin{enumerate}[(i)]
    \item Each vertex in $H_M$ has outdegree exactly one.
    \item Each component of $H_M$ contains a single simple cycle.
\end{enumerate}
\end{lemma}

\begin{proof}
(i) is direct from the definition of a switching graph. (ii) No vertex points to itself, so there is no self loop in $H_M$. If there is no cycle in one component, then there exists at least one sink vertex (consider the topological sort of $H_M$), contradicting (i). If there are two cycles in one component, consider any path that connects these two cycles. There must be a vertex with at least two outgoing edges again contradicting (i).
\end{proof}

From Definition \ref{def:rotation}, it is easy to see that any such simple cycle defines the men in a rotation exposed in $M$, in the order that they appear in the rotation. On the other hand, based on the uniqueness of $next_M(m)$ for each $m \in D$, if $m$ belongs to some rotation $\rho$, e.g. $m = m_i$, then $m_{i+1}$ is uniquely determined, that is $next_M(m)$. Hence the men in $\rho$ must be a simple cycle in $H_M$. 

We know from Section \ref{sec:pseudoforest} that every cycle in $H_M$ can be found in NC. It is obvious that the elimination of a rotation can be done in one parallel step. Thus, we are left to show that $H_M$ can be constructed in NC. 

Let us assume that a stable marriage instance is described by the sets of preference lists, represented as matrices $mp$ and $wp$ defined by
\begin{itemize}
    \item $mp[m,i] = w$ if woman $w$ is ranked of $i$ in $m$'s preference list
    \item $wp[w,i] = m$ if man $m$ is ranked of $i$ in $w$'s preference list
\end{itemize}

We also define the ranking matrices $mr$ and $wr$ as below
\begin{itemize}
    \item $mr[m,w] = i$ if woman $w$ is ranked of $i$ in $m$'s preference list
    \item $wr[w,m] = i$ if man $m$ is ranked of $i$ in $w$'s preference list
\end{itemize}

We need to identify $s_M(m)$ and $next_M(m)$ for each man $m$. Suppose for each woman $w$ we delete all pairs $(m', w)$ such that $w$ prefers $p_M(w)$ to $m'$. In the resulting preference lists, which we call reduced lists, $p_M(w)$ is the last entry in $w$'s list, and $p_M(m)$ is the first entry in $m$'s list for if any woman $w'$ remains above $p_M(m)$, then $(m, w')$ blocks $M$. Moreover, $s_M(m)$ is the second entry in $m$'s list if exists, for by definition, it is the highest ranked woman $w$ on $m$'s list such that $w$ prefers $m$ to $p_M(w)$. $next_M(m)$ is simply the partner in $M$ of woman $s_M(m)$. 

From the algorithmic aspect, for each entry $(m,w)$ to be deleted in parallel, we call the ranking matrix $mr$ to obtain woman $w$'s rank on $m$'s list. Then call the preference matrix $mp$ and use soft-deletion, i.e. mark the entry $mp[m,mr[m,w]]$ zero. After each entry is soft-deleted, we can compress the preference list using parallel prefix sum technique. The resulting preference lists are reduced lists. Now we obtain all pairs $(m, next_M(m))$ and it is easy to construct $H_M$.

\subsection{Example of Stable Matchings}

Figure \ref{fig:stable-matching} is an example of a stable marriage instance. The reader can verify that the matching $M$ denoted by underlining is stable.

\begin{figure}[h]
    \centering
        \begin{subfigure}[b]{0.45\textwidth}
        \centering
        $\begin{aligned}
            m_1: &\sh w_5\sh w_7\sh w_1\sh w_2\sh w_6\sh \underline{w_8}\sh w_4\sh w_3\\
            m_2: &\sh w_2\sh \underline{w_3}\sh w_7\sh w_5\sh w_4\sh w_1\sh w_8\sh w_6\\
            m_3: &\sh w_8\sh \underline{w_5}\sh w_1\sh w_4\sh w_6\sh w_2\sh w_3\sh w_7\\
            m_4: &\sh w_3\sh w_2\sh w_7\sh w_4\sh w_1\sh \underline{w_6}\sh w_8\sh w_5\\
            m_5: &\sh \underline{w_7}\sh w_2\sh w_5\sh w_1\sh w_3\sh w_6\sh w_8\sh w_4\\
            m_6: &\sh \underline{w_1}\sh w_6\sh w_7\sh w_5\sh w_8\sh w_4\sh w_2\sh w_3\\
            m_7: &\sh \underline{w_2}\sh w_5\sh w_7\sh w_6\sh w_3\sh w_4\sh w_8\sh w_1\\
            m_8: &\sh w_3\sh w_8\sh \underline{w_4}\sh w_5\sh w_7\sh w_2\sh w_6\sh w_1\\
        \end{aligned}$
        \caption{Men's preferences}
    \end{subfigure}
    ~
    \begin{subfigure}[b]{0.45\textwidth}
        \centering
        $\begin{aligned}
            w_1: &\sh m_5\sh m_3\sh m_7\sh \underline{m_6}\sh m_1\sh m_2\sh m_8\sh m_4\\
            w_2: &\sh m_8\sh m_6\sh m_3\sh m_5\sh \underline{m_7}\sh m_2\sh m_1\sh m_4\\
            w_3: &\sh m_1\sh m_5\sh m_6\sh \underline{m_2}\sh m_4\sh m_8\sh m_7\sh m_3\\
            w_4: &\sh \underline{m_8}\sh m_7\sh m_3\sh m_2\sh m_4\sh m_1\sh m_5\sh m_6\\
            w_5: &\sh m_6\sh m_4\sh m_7\sh \underline{m_3}\sh m_8\sh m_1\sh m_2\sh m_5\\
            w_6: &\sh m_2\sh m_8\sh m_5\sh m_3\sh \underline{m_4}\sh m_6\sh m_7\sh m_1\\
            w_7: &\sh m_7\sh \underline{m_5}\sh m_2\sh m_1\sh m_8\sh m_6\sh m_4\sh m_3\\
            w_8: &\sh m_7\sh m_4\sh \underline{m_1}\sh m_5\sh m_2\sh m_3\sh m_6\sh m_8\\
        \end{aligned}$
        \caption{Women's preferences}
    \end{subfigure}
    \caption{The stable marriage instance of size 8 and the stable matching $M$ denoted by underlining}
    \label{fig:stable-matching}
\end{figure}

Figure \ref{fig:reduced-stable} shows the reduced lists of the men for the stable matching $M$. The second column corresponds to $s_M(m)$ for each $m$.
\begin{figure}[h]
    \centering
        $\begin{aligned}
        m_1: &\sh \underline{w_8} \sh w_3\\
        m_2: &\sh \underline{w_3}\sh w_6\\
        m_3: &\sh \underline{w_5}\sh w_1\sh w_6\sh w_2\\
        m_4: &\sh \underline{w_6}\sh w_8\sh w_5\\
        m_5: &\sh \underline{w_7}\sh w_2\sh w_1\sh w_3\sh w_6\\
        m_6: &\sh \underline{w_1}\sh w_5\sh w_2\sh w_3\\
        m_7: &\sh \underline{w_2}\sh w_5\sh w_7\sh w_8\sh w_1\\
        m_8: &\sh \underline{w_4}\sh w_2\sh w_6\\
        \end{aligned}$
    \caption{The reduced lists of the men for the stable matching $M$}
    \label{fig:reduced-stable}
\end{figure}

Finally we give the switching graph $H_M$ for the stable matching $M$ in Figure \ref{fig:stable-switching}.

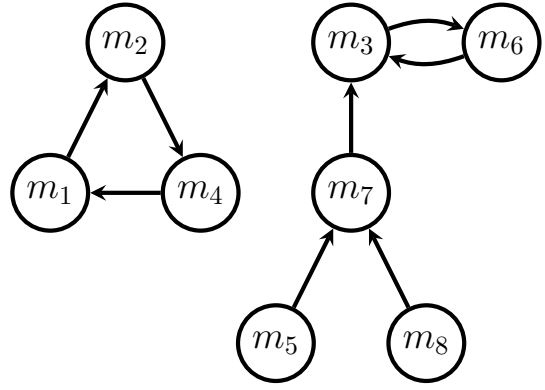
\begin{figure}[ht]
    \centering
    \begin{tikzpicture}[> = stealth, auto, node distance=3cm, ultra thick,
   node_style/.style={circle,draw=black,font=\sffamily\Large\bfseries},
   edge_style/.style={draw=black, ultra thick}]
    \node[node_style] (m1) at (0,0) {$m_1$};
    \node[node_style] (m2) at (1,2) {$m_2$};
    \node[node_style] (m3) at (4,2) {$m_3$};
    \node[node_style] (m4) at (2,0) {$m_4$};
    \node[node_style] (m5) at (3,-2) {$m_5$};
    \node[node_style] (m6) at (6,2) {$m_6$};
    \node[node_style] (m7) at (4,0) {$m_7$};
    \node[node_style] (m8) at (5,-2) {$m_8$};
    \path[->] (m1) edge node {} (m2);
    \path[->] (m2) edge node {} (m4);
    \path[->] (m4) edge node {} (m1);
    \path[->, bend left=20] (m3) edge node {} (m6);
    \path[->, bend left=20] (m6) edge node {} (m3);
    \path[->] (m7) edge node {} (m3);
    \path[->] (m5) edge node {} (m7);
    \path[->] (m8) edge node {} (m7);
    \end{tikzpicture}
    \caption{The switching graph $H_M$}
    \label{fig:stable-switching}
\end{figure}
\section{Summary and Open Problems}
This paper has established the result that the popular matching problem without ties is in NC. The notion of pseudoforest may have other applications for designing parallel algorithms.

We have shown that maximum-cardinality bipartite matching is NC-reducible to popular matching. One open problem is Conjecture \ref{conj:nc-reduction-left}. If it is true, then it means these two problems are NC-equivalent. The other open problem is establishing the NC reduction among several other matching problems in preference systems such as Pareto-optimal matching and rank-maximal matching. Another open problem is to determine if there is an RNC algorithm for popular matching problem with ties.

\bibliographystyle{IEEEtran}
\bibliography{main}


\end{document}